\documentclass[a4paper]{article}

\usepackage{amsmath,amsthm,amssymb}

\usepackage{xcolor}
\usepackage{listings}
\definecolor{mGreen}{rgb}{0,0.6,0}
\definecolor{mGray}{rgb}{0.5,0.5,0.5}
\definecolor{mPurple}{rgb}{0.58,0,0.82}
\definecolor{backgroundColour}{rgb}{0.95,0.95,0.92}
\lstdefinestyle{CStyle}{
    backgroundcolor=\color{backgroundColour},
    commentstyle=\color{mGreen},
    keywordstyle=\color{blue},
    numberstyle=\tiny\color{mGray},
    stringstyle=\color{mPurple},
    basicstyle=\footnotesize,
    breakatwhitespace=false,
    breaklines=true,
    captionpos=top,
    keepspaces=true,
    numbers=left,
    numbersep=5pt,
    showspaces=false,
    showstringspaces=false,
    showtabs=false,
    tabsize=2,
    language=C
}

\let\ts\textstyle

\def\R{\mathbb{R}}

\def\h#1{\hskip#1pt}
\def\m#1{\mskip#1mu}
\def\md{\m3|\m3}

\def\maj{\mathop{\rm maj}\nolimits}
\def\frc#1#2{{\ts{#1\over#2}}}
\def\cost{\mathop{\rm cost}\nolimits}
\def\para#1.{\ifdim\lastskip<0pt\relax\else\removelastskip\fi%
	\bigbreak\penalty-250\par\noindent{\bf#1.\h4}}
\def\th#1#2{T^{#1}_{#2}}
\def\qr{Q_\mathrm R}
\def\f#1{F^{#1}_{k,n}}

\def\summ{\textstyle\sum\nolimits}

\def\acases{\left.\begin{array}{ll}}
\def\bcases{\left\{\begin{array}{ll}}
\def\ecases{\end{array}\right.}
\def\qand{\quad\text{and}\quad}
\newtheorem{theorem}{Theorem}%<<<
\newtheorem{lemma}[theorem]{Lemma}
\newtheorem{corollary}[theorem]{Corollary}
\newtheorem{proposition}[theorem]{Proposition}

\newtheorem{definition}{Definition}
\newtheorem*{remark}{Remark}

\begin{document}

% Title and abstract <<<
\title{Lower bounds for uniform read-once threshold formulae in the randomized decision tree model}
%\date{\vspace{-5ex}}
\author{Nikos Leonardos%
	\thanks{
		This research is co-financed by Greece and the European Union (European
		Social Fund---ESF) through the Operational Programme ``Human Resources
		Development, Education and Lifelong Learning'' in the context of the
		project ``Reinforcement of Postdoctoral Researchers---2nd Cycle''
		(MIS-5033021), implemented by the State Scholarships Foundation (IKY).
	}
	\thanks{nikos.leonardos@gmail.com}
}
\maketitle

\begin{abstract} 
	We investigate the randomized decision tree complexity of a specific class
	of read-once threshold functions. 
	A read-once threshold formula can be defined by a rooted tree, every
	internal node of which is labeled by a threshold function $\th nk$ (with
	output 1 only when at least $k$ out of $n$ input bits are 1) and each leaf
	by a distinct variable. Such a tree defines a Boolean function in a natural
	way. We focus on the randomized decision tree complexity of such functions,
	when the underlying tree is a uniform tree with all its
	internal nodes labeled by the same threshold function.
	We prove lower bounds of the form $c(k,n)^d$, where $d$ is the depth of the
	tree.
	We also treat trees with alternating levels of AND and OR gates separately
	and show asymptotically optimal bounds, extending the known bounds for the
	binary case.
\end{abstract}%>>>

% Section 1
\section{Introduction}%<<<
Boolean decision trees constitute one of the simplest computational models.
It is therefore intriguing when the complexity of a function is still unknown.
A notable example is the recursive majority-of-three function.
This function can be represented by a uniform ternary tree of depth $d$,
such that every internal node has three children and all leaves are on the
same level. The function computed by interpreting the tree as a circuit with
internal nodes labeled by majority gates (with output 1 only when at least two
of the three inputs are 1) is $\maj_d$, the recursive majority-of-three of
depth $d$.

This function seems to have been given by Ravi Boppana (see Example~$1.2$ in
\cite{sw86}) as an example of a function that has deterministic complexity
$3^d$, while its randomized complexity is asymptotically smaller. 
Other functions with this property are known. Another notable example is the
function $\mathrm{nand}_d$, first analyzed by Snir \cite{snir85}. This is the
function represented by a uniform binary tree of depth $d$, with the internal
nodes labeled by $\mathrm{nand}$ gates. Equivalently, the internal nodes can
be labeled by AND and OR gates, alternately at each level.

A simple randomized framework that can
be used to compute both $\maj_d$ and $\mathrm{nand}_d$ is the following. Start
at the root and, as long as the output is not known, choose a child at random
and evaluate it recursively. Algorithms of this type are called {\em
directional}. For $\maj_d$ the directional algorithm computes the output
in $(8/3)^d$ queries. It was noted by Boppana and also in \cite{sw86} that
better algorithms exist for $\maj_d$.
(See \cite{amano2010} for a more recent study of directional algorithms.)
Interestingly, Saks and Wigderson show
that the directional algorithm is optimal for the $\mathrm{nand}_d$ function,
and show that its zero-error randomized decision tree complexity is 
$\Theta\bigl(\smash{(\frc{1+\sqrt33}4)}^d\bigr)$. 
Their proof uses a bottom up induction and generalized costs. Their method of
generalized costs allows them to charge for a query according to the value of
the variable. 
In this work we use this method to show that the directional algorithm is
optimal for uniform AND-OR trees where each gate receives $n$ bits. 

More generally, there have been studied functions that can be represented by
formulae involving threshold functions as connectives.
A threshold $k$-out-of-$n$ function, denoted $\th nk$, is a Boolean function
of $n$ arguments that has value 1 if at least $k$ of the $n$ Boolean input
values are 1. A threshold formula can be defined as a rooted tree with labeled
nodes; each internal node is labeled by a threshold function and each leaf by
a variable. If each variable appears exactly once the formula is called
\emph{read-once};
if only AND and OR connectives appear, it is called a \emph{Boolean} read-once
formula.
A formula represents a Boolean function in a natural way:
the tree is evaluated as a circuit where each internal node is the
corresponding threshold gate. If the formula is read-once, the function it
represents is also called read-once.
If no OR gate is an input to another OR gate, and the same for AND gates, then
the formula is non-degenerate (see Theorem~2.2 in \cite{hnw93}) and uniquely
represents the corresponding function.
Thus, we may define the depth of $f$ as the maximum depth of
a leaf in the unique tree-representation. Note also that, for read-once
formulae, the set of variables coincides with the set of leaves in their
representation as a tree.
We prove lower bounds for the subclass that is represented by uniform trees
(full and complete trees with all their leaves at the same level and with each
internal node having the same number of children).
%>>>
\paragraph{Related work and our results.}%<<<
Heiman and Wigderson \cite{hw91} managed to show that for every Boolean
read-once function $f$ we have $R(f)\in\Omega(D(f)^{0.51})$, where $R(f)$ and
$D(f)$ are the randomized and deterministic complexity of $f$ respectively.
For the subclass of such functions that their representing tree is binary, the
question was investigated more recently in \cite{amano2011}.
Heiman, Newman, and Wigderson \cite{hnw93} showed that read-once formulae with
threshold gates have zero-error randomized complexity $\Omega(N/2^d)$ (here
$N$ is the number of variables and $d$ the depth of a canonical
tree-representation of the read-once function).

Santha showed in \cite{santha95} that when considering a notion of balanced
read-once formulae with AND and OR gates, the zero-error and bounded-error
complexities are intimately related. This holds, in particular, for the class
of uniform formulae that we consider. Furthermore, it is not hard to verify
that our lower bounds for uniform threshold read-once formulae also hold for
bounded-error algorithms if multiplied by $(1-2\delta)$, where $\delta$ is the
error allowed. Consequently, we focus on zero-error complexity in this work.

With respect to specific read-once functions,
in contrast to the exact asymptotic bounds we have for $\mathrm{nand}_d$
(also in the quantum model \cite{bs04}) there
had been no progress on the randomized decision tree complexity of $\maj_d$ for
several years.
However, recent papers have narrowed the gap between the upper and lower bounds
for recursive majority.
An $\Omega\bigl((7/3)^d\bigr)$ lower bound was showed in
\cite{jks03}. Jayram, Kumar, and Sivakumar, proved their bound using tools
from information theory and a top down induction. 
Furthermore, they presented a non-directional algorithm that improves the 
$O\bigl((8/3)^d\bigr)$ upper bound.
Landau, Nachmias, Peres, and Vanniasegaram~\cite{lns2006}, showed how to remove
the information theoretic notions from the proof in \cite{jks03}, keeping its
underlying structure the same.
Magniez, Nayak, Santha, and Xiao \cite{mnsx11}, significantly
improved the lower bound to $\Omega\bigl((5/2)^d\bigr)$
and the upper bound to $O(2.64946^d)$.
Subsequently, in \cite{leonardos2013}, the lower bound was further improved to 
$\Omega(2.55^d)$, building upon the techniques of \cite{sw86}.
The bound was further improved with a computer-assisted proof in \cite{mnsstx2016}.
The currently known best lower bound is $\Omega(2.59^d)$ from \cite{gj2016}.

The above methods work by reducing the complexity of $\maj_d$ to a problem of
constant size. 
It is interesting to note that when this constant-size problem relates to 
$\maj_1$ (i.e., the majority of three bits), all these methods obtain the
$2.5^d$ lower bound. Furthermore, in order to make further progress,
analysis of analogous problems of greater depth is required. In particular,
in \cite{leonardos2013}, an intermediate problem between $\maj_1$ and $\maj_2$
is constructed and analyzed with a reasonable case analysis (without the need
of computer assistance) to obtain a $2.55^d$ lower bound;
in \cite{mnsstx2016}, the $2.55^d$ bound is surpassed by solving depth-3
and depth-4 problems using a computer;
finally, \cite{gj2016} obtains the currently best bound via the solution of
a linear program (obtained using a computer) which relates to $\maj_3$.
Pushing any of the above further seems computationally intractable and will
yield possibly modest improvements. We believe a human-readable solution can
shed more light into the problem. 

In this work we focus our attention to threshold read-once functions 
which are represented by uniform trees and each internal node is
labeled by the same threshold gate. 
Let us denote by $F^d_{\land,n}$ (resp.~$F^d_{\lor,n}$) the function
represented by a uniform tree of depth $d$ with AND and OR gates of fan-in $n$
alternating with each level and the
root labeled by an AND gate (resp.~an OR gate). 
Denote by $F^d_{k,n}$ the function represented by a uniform
tree of depth $d$ with each gate being $T_k^n$.
With respect to these classes of read-once functions we prove the following
theorems in the model of randomized decision trees.
(See Section~2 for relevant definitions.)

\begin{theorem}\label{thm:andor}
	The directional algorithm is optimal for 
	$F^d_{\wedge,n}$ and $F^d_{\vee,n}$. In particular, for even $d$,
	\\[-10pt]
	\[
		R(F^{d}_{\wedge,n}),R(F^d_{\vee,n})=\Theta\Biggl(\Biggl[
			%(n-1)\Biggl(\frac14+\sqrt{\frac12+\frac{16n}{(n-1)^2}}\Biggr)
			\frac{n-1}4\Biggl(1+\sqrt{1+\frac{16n}{(n-1)^2}}\Biggr)
			%\frac{n-1+\sqrt{n^2+14n+1}}4
			%n+\frac{(n-1)^2}8+\frac{(n-1)^2}8\sqrt{1+\frac{16n}{(n-1)^2}}
			%\Bigl(\frac{n+1}2\Bigr)^2+\frac{2n}{1+\sqrt{1+\frac{16n}{(n-1)^2}}}
		\Biggr]^{d}\Biggr).\]
		%	\Bigl(\frac{n+1}2\Bigr)^2+\frac{2n}{1+\sqrt{1+\frac{16n}{(n-1)^2}}}
		%\biggr)^{d/2}\Biggr).\]
\end{theorem}
\begin{theorem}\label{thm:directional}
	For $1<k<n$,
	\[
		R(F^d_{k,n})=O\Biggl(\Biggl[
			%n-\frac{(k-1)(n-k)(n+2)}{(n-1)(n-k+2)(k+1)}
			%\frac{n+1}2+\frac{n-1}2\sqrt{1-\frac{4(k-1)(n-k)(n+2)}{(k+1)(n-k+2)(n-1)^2}}
			\frac{n+1}2\Biggl(1+\sqrt{1-\frac{8k(n-k+1)}{(n-k+2)(k+1)(n+1)}}
		\Biggr)\Biggr]^d\Biggr)
	.\]
\end{theorem}
\begin{theorem}\label{thm:main}
	For $1<k<n$,
	\[
		R(F^d_{k,n})=\Omega\Biggl(\Biggl[
			%\frac{n+1}2+\frac{n+1}2\sqrt{1-\frac{n+3kn-3k^2+3k}{(n+1)^2}}
			\frac{n+1}2\Biggl(1+\sqrt{1-\frac{3k(n-k+1)+n}{(n+1)^2}}
		\Biggr)\Biggr]^d\Biggr).\]
\end{theorem}

Note that the last two expressions are symmetric in $k$ around $(n+1)/2$,
reflecting the fact that $R(F^d_{k,n})=R(F^d_{n+1-k,n})$ (as can be seen by
negating and pushing the negations down to the leaves).
Setting $n=2$ in Theorem~\ref{thm:andor} we retrieve the bound of Saks and
Wigderson for $\mathrm{nand}_d$.
Similarly, for $k=2$ and $n=3$, Theorems~\ref{thm:directional} and
\ref{thm:main} give known bounds for $\maj_d$: the $(8/3)^d$ upper bound of
the directional algorithm and the $2.5^d$ lower bound respectively. 
We expect neither the upper bound nor the lower bound to be optimal.
Our interest is mostly in Theorem~\ref{thm:main} and also
Theorem~\ref{thm:andor}.
Theorem~\ref{thm:directional} is obtained by a straightforward analysis of the
directional algorithm and is stated mostly for reference and comparison with
the bound of Theorem~\ref{thm:main}.

We note that one could obtain similar bounds using the methodology of
\cite{jks03,mnsstx2016} or \cite{gj2016}. In particular, formal descriptions
of feasible solutions to the corresponding linear programs defined in
\cite{gj2016}, can be used along with their Composition Theorem to obtain
lower bounds for the class of functions we consider.
%>>>

% Section 2
\section{Definitions and notation}%<<<
In this section we introduce basic concepts related to decision tree
complexity and the tree-functions that we consider.
The reader can find a more complete exposition of decision tree complexity in
the survey of Buhrman and de~Wolf~\cite{bw2002}.

%>>>
\subsection{Definitions pertaining to trees}%<<<
For a rooted tree $T$,
the {\em depth} of a node is the number of edges on the path to the root; 
the {\em depth} of the tree is the maximum depth of a leaf;
the {\em level} $i$ of a tree consists of all nodes of depth $i$;
the {\em height} of a node is the number of edges on the longest path between
the node and any descendant leaf.
We call a tree {\em uniform} if all the leaves are on the same level and all
internal nodes have the same number of children.

Consider the uniform tree of depth $d$ in which every internal node has $n$
children and is labeled by $\th nk$ (the $k$-out-of-$n$ threshold gate) with
$1<k<n$. 
We denote both the tree and the corresponding read-once threshold function 
by the same symbol $\f d$ as it will be clear from the context what it
refers to.

For the case of AND or OR gates, we consider uniform trees with the root
labeled by an AND (resp.\ OR) gate and subsequent levels of the tree are
labeled by OR (resp.\ AND) and AND (resp.\ OR) gates alternately. When each
gate receives $n$ inputs, we denote
these trees by $F_{\wedge,n}^d$ and $F_{\vee,n}^d$, when the root is labeled
by an AND and OR respectively. 

The inputs considered hard for the functions we consider are the
\emph{reluctant} inputs (\cite{sw86}).
Call an input to a $T_k^n$-gate {\em reluctant}, if it has exactly $k$ or
$k-1$ ones.
Call an input to a threshold read-once formula {\em reluctant}, if it is such
that the input to every gate is reluctant.
A {\em reluctant distribution} has only reluctant inputs in its support.

We will consider trees that have all leaves in the two greatest levels. 
We now define the reluctant distribution over such trees that we will work
with.
We focus on {\em positive inputs} only; i.e., inputs that evaluate to true.
The definition is recursive. If the tree consists of a single variable, then
all the mass in on 1. Otherwise, consider a tree $T$ of depth $d$ with all
leaves in levels $d-1$ and $d$ (level $d-1$ might not have any leaves if $T$
is uniform). Obtain a tree $T'$ by removing all the
children of an internal node $u$ with label $\th nk$ in level $d-1$.
Let $\mu'$ denote the reluctant distribution for $T'$.
Consider any reluctant input $x$ to $T$ and let $b$ be the value that $u$ (as
a gate) outputs.
Let $x'$ be the corresponding (reluctant) input to $T'$; i.e., $x'$ assigns
$b$ to $u$ and agrees with $x$ everywhere else.
Define $\mu(x)=\mu'(x')/{n\choose k}$ if $b=1$ and $\mu(x)=\mu'(x')/{n\choose
k-1}$ if $b=0$. The reluctant distribution for $F{}^d_{k,n},$
$F{}^d_{\wedge,n}$, and
$F{}^d_{\vee,n}$, is denoted $\mu_{k,n}^d$, $\mu_{\wedge,n}^d$, and
$\mu_{\vee,n}^d$, respectively.
Note that although we will restrict the analysis over positive inputs, the
algorithms are still required to be zero-error over any input.
%>>>
\subsection{Definitions pertaining to decision trees}%<<<
A {\em deterministic Boolean decision tree} $Q$ over $\{0,1\}^n$
is a rooted and ordered binary tree.
Each internal node is labeled by
$i\in[n]=\{1,2,\dots,n\}$
and each leaf by a value from $\{0,1\}$.
An {\em assignment} or an {\em input} is a member of $\{0,1\}^n$.
The output $Q(x)$ of $Q$ on an input $x$ is defined recursively
as follows. Start at the root and let its label be $i$. If
$x_i=0$, we continue with the left child of the root; if $x_i=1$, we
continue with the right child of the root. We continue recursively until we
reach a leaf. We define $Q(x)$ to be the label of that leaf. When we reach
an internal node, we say that $Q$ {\em queries} 
the corresponding variable and {\em reads} its value.
We say that $Q$ {\em computes} a Boolean function
$f:\{0,1\}^n\to\{0,1\}$, if for all $x\in\{0,1\}^n$,
$Q(x)=f(x)$.
The {\em cost of $Q$ on input $x$},
$\cost(Q;x)$, is the number of variables queried when the input is
$x$. The {\em cost of $Q$}, $\cost(Q)$, is its {\em depth}, the maximum
distance of a leaf from the root.
The {\em deterministic complexity}, $D(f)$, of a Boolean function $f$ is the
minimum cost over all Boolean decision trees that compute $f$.

A {\em randomized Boolean decision tree} $Q_{\mathrm R}$ is a distribution $p$ over
deterministic decision trees. On input $x$, a deterministic decision tree
is chosen according to $p$ and evaluated. The {\em cost of $Q_{\mathrm R}$ on
input $x$} is $\cost(Q_{\mathrm R};x)=\summ_{Q}p(Q)\cost(Q;x)$.
The {\em cost} of $Q_{\mathrm R}$, $\cost(Q_{\mathrm R})$, is
$\max_x\cost(Q_{\mathrm R};x)$. 
A randomized decision tree $Q_{\mathrm R}$ {\em computes} a Boolean function
$f$ (with zero error), if $p(Q)>0$ only when $Q$ computes $f$. 
The {\em randomized complexity}, $R(f)$, of a Boolean function $f$ is the
minimum cost over all randomized Boolean decision trees that compute $f$.

We are going to take a distributional view on randomized algorithms. Let $\mu$
be a distribution over $\{0,1\}^n$ and $Q_{\mathrm R}$ a randomized decision tree.
The {\em expected cost of $Q_{\mathrm R}$ under $\mu$} is 
$\cost_\mu(Q_{\mathrm R})=\summ_x\mu(x)\cost(Q_{\mathrm R};x).$
The {\em expected complexity under $\mu$}, $R_\mu(f)$, of a
Boolean function $f$ is the minimum expected cost under $\mu$ of any
randomized Boolean decision tree that computes $f$.
Clearly, $R(f)\ge R_\mu(f)$, for any $\mu$, and thus we can prove lower bounds
on randomized complexity by providing lower bounds for the expected
cost under any chosen distribution.
%>>>

% Section 3
\section{The method of generalized costs for read-once threshold formulae}%<<<

Our goal is to prove a lower bound on the expected cost (under the reluctant
distribution we defined in the previous section) of any randomized decision
tree $Q_{\mathrm R}$ that computes a uniform threshold read-once function. The
high-level outline of our proof is as follows.
Given any decision tree $Q_{\mathrm R}$ computing the function $F^d_{k,n}$, we
define a randomized decision tree $Q_{\mathrm R}'$ that computes $\f{d-1}$.
The crucial part is to show that we may charge $Q_{\mathrm R}'$ more for each
query, while maintaining that the cost of $\qr'$ is upper bounded by that of
$\qr$.
Applying this step repeatedly reduces our task to determining the complexity
of a function on a single variable with a high cost to query it. To implement
this plan we utilize the method of generalized costs of Saks and Wigderson
\cite{sw86}. 

The reduction from $F^d_{k,n}$ to $F^{d-1}_{k,n}$ proceeds inductively and in
several steps. 
In each step we shrink to a leaf the parent of $n$ leaves that are siblings.
That is, we start with a function $F(x,y)$, where $x\in\{0,1\}^n$ (recall that
$n$ is the fan-in of $\th kn$), and reduce
its complexity to a function $F'(u,y)$, where $u\in\{0,1\}$, such that
$F(x,y)=F'(\th nk(x),y)$.
Since we only need one query to discover the value of $u$ in $F'$, but at least $k$
in $F$, we should be able to charge more an algorithm computing $F'$ for
querying $u$, than an algorithm computing $F$ for querying any of the bits in
$x$. This is captured by the introduction of cost functions, discussed next.
Continuing this way, all nodes in level $d-1$ will become leaves and the
uniform tree of depth $d$ that we started with has been transformed to one of
depth $d-1$.
%>>>
\subsection{Cost functions}%<<<
Define a {\em cost-function} over $\{0,1\}^N$
to be a function $\phi:\{0,1\}^N\times[N]\to\R$. We extend
the previous cost-related definitions as follows.
The cost of a decision tree $Q$ under cost-function $\phi$ on input
$x\in\{0,1\}^N$ is
\[
	\cost(Q;\phi;x)=\summ_{i\in S}\phi(x;i),\quad
	\hbox{where $S=\{i\md\hbox{$i$ is queried by $Q$ on input $x$}\}$.}
\]
The cost of a randomized decision tree $Q_{\mathrm R}$ on input $x$ under
cost-function $\phi$ is
\[
	\cost(Q_{\mathrm R};\phi;x)=\summ_Qp(Q)\cost(Q;\phi;x)
,\]
where $p$ is the corresponding distribution over deterministic decision trees.
Finally, the expected cost of a randomized decision tree $Q_{\mathrm R}$ under
cost-function $\phi$ and distribution $\mu$ is 
\[
	\cost_\mu(Q_{\mathrm R};\phi)=\summ_x\mu(x)\cost(Q_{\mathrm R};\phi;x)
.\]

We will mostly work with simple cost functions (as was also the
case in \cite{sw86}), in which $\phi(x,i)$ depends only on $x_i$ and $i$. In
particular, we may express such a function with a pair $c=(c_0,c_1)$, where
$c_0,c_1$ are functions mapping each variable to a real number and
$c_{x_i}(i)=\phi(x,i)$. We call such cost functions \emph{local}.
Defining $\cost(Q;c;x)=\summ_{i\in S}c_{x_i}(i)$, for a local cost-function
$c$, we use both expressions 
interchangeably in the above
quantities.
%>>>
\subsection{The method and some preliminary definitions} %<<<
Consider a tree $T$ of depth $d+1$ such that all internal nodes have degree
$n$ and all leaves are on levels $d$ and $d+1$. If we treat every internal
node as a $\th nk$-gate, this tree represents a function $F$ and has an
associated reluctant distribution $\mu$.
Suppose $Q$ is a randomized decision tree that computes $F$ and
$c=(c_0,c_1)$ a local cost-function.
We define a process that shrinks $T$ to a smaller tree $T'$ (of depth $d$ or
$d+1$) and also a corresponding randomized decision tree $Q'$ that
computes the function $F'$ represented by $T'$.
The crucial part is to show that for a ``more expensive'' local cost-function
$c'$ and the reluctant distribution $\mu'$ for $F'$,
\begin{equation}\label{eq:costineq}
	\cost_\mu(Q;c)\ge\cost_{\mu'}(Q';c')
.\end{equation}

The main ingredient in this framework is the shrinking process. It entails
removing $n$ leaves, which we identify with $[n]$,
so that their parent $u$ would become a leaf in $T'$. 
Given an algorithm $Q$ for $F$ we obtain an algorithm $Q'$ for $F'$ as
follows. 
On input $(x_u,y)$, $Q'$ first chooses uniformly at random 
$z\in{[n-1]\choose k-1}$ and $i\in[n]$ and then simulates $Q$ on
$x=(z_1,\dots,z_{i-1},x_u,z_i,\dots,z_{n-1},y)$.
Clearly, if $Q$ computes $F$, then $Q'$ is a randomized decision tree that
computes $F'$.

Our goal is to determine the ``most expensive'' cost function $c'$ for $T'$
for which we can argue (\ref{eq:costineq}).
To that end, it will be useful to express $\cost_{\mu'}(Q';c')$ in terms of
$Q$ and $\mu$. 
For an input $x$ for $F$ and a leaf $i$ of $T$, let $b=\th nk(x_1,\dots,x_n)$
and define a cost function $\psi$ for $F$ as follows. 
\begin{equation}\label{equ.psistar}
	\psi(x;i)=\bcases
		c'_{x_i}(i),			&\text{if $i\notin[n]$};\cr
		c'_1(u)/k, 				&\text{if $i\in[n]$ and $x_i=b=1$};\cr
		c'_0(u)/(n-k+1), 	&\text{if $i\in[n]$ and $x_i=b=0$};\cr
		0, 								&\text{if $i\in[n]$ and $x_i\ne b$}.\cr
	\ecases
\end{equation}
Informally, $\psi$ agrees with $c'$ outside $[n]$, is zero when $Q'$ will never
query $i$, and in the rest of the cases has value so that
the following proposition holds.
\begin{proposition}\label{prop:costdiff}
	$\cost_{\mu'}(Q';c')=\cost_\mu(Q;\psi)$.
\end{proposition}
\begin{proof}
	It suffices to prove the equality for a deterministic $Q$.
	The expected cost of any variable outside $[n]$ is clearly the same for
	both algorithms. Thus, it suffices to show that the contribution of $u$ to
	the cost of $Q'$ equals that of the variables in $[n]$ to the cost of $Q$. 
	Intuitively, a query of $Q$ to a variable $i\in[n]$ such that $x_i=b$, leads
	to a query of $Q'$ to $u$ with probability $1/k$ or $1/(n-k+1)$, according
	to $b=1$ or $b=0$ respectively. More formally,
	let $\chi_i(x)$ be the indicator function of the event that $Q$ queries $i$
	when the input is $x$.
	We consider the inputs with $x_u=1$.

	Recalling the recursive distribution of the reluctant distribution $\mu$,
	the contribution of the leaves in $[n]$ to the cost of $Q$ can be written
	\[
		\sum_{y}\sum_{z,i}\frac{\mu'(x_u,y)}{{n\choose k}}
			\cdot\chi_i(x)\cdot\frac{c_1'(u)}k
	.\]
	In the expression above, $x=(z_1,\dots,z_{i-1},x_u,z_i,\dots,z_{n-1},y)$, as
	in the description of $Q'$. Note that although each $x$ is encountered $k$
	times, each $i\in[n]$ such that $z_i=1$ contributes exactly once (those with
	$z_i=0$ contribute 0 anyway).
	The contribution of $u$ to the cost of $Q'$ is
	\[
		\sum_{y}\mu(x_u,y)\sum_{z,i}\frac1{n{n-1\choose k-1}}
			\cdot\chi_i(x)\cdot{c_1'(u)}
	,\]
	since there are $n$ choices for $i$ and ${n-1\choose k-1}$ choices for $z$.
	Recalling ${n\choose k}=\frac nk{n-1\choose k-1}$,
	the equality of the two expressions follows.

	The inputs with $x_u=0$ can be analysed in the same way, substituting $c'_0$
	for $c'_1$ in both expressions and $n-k+1$ for $k$ in the first one.
\end{proof}

Let us go back to our plan: given a local cost-function $c$ for $F$, determine
the most expensive local cost-function $c'$ for $F'$ for which we can argue
(\ref{eq:costineq}).
We will work with local cost-functions that charge the same value for
discovering a 0 and the same for discovering a 1 at a given level.
That is, if $i$ and $j$ are leaves at the same level, then $c_0(i)=c_0(j)$ 
and $c_1(i)=c_1(j)$.

We shall define $c'$ so that it agrees with $c$ for all $i\notin[n]$.
Defining a cost-function $\phi(x;i)=c_{x_i}(i)-\psi(x;i)$ for $i\in[n]$ and
$\phi(x;i)=0$ otherwise (where $\psi$ is the one obtained from 
Proposition~\ref{prop:costdiff}), we have
\begin{equation}\label{eq:phi}
	\cost_{\mu}(Q;c)\ge\cost_{\mu'}(Q';c')\iff\cost_\mu(Q;\phi)\ge0
.\end{equation}
Clearly, we only need to verify $\cost_\mu(Q;\phi)\ge0$
with respect to the leaves in $[n]$.
Thus, we are led to study decision trees over ${[n]\choose k}$ that are
non-empty (i.e., query at least one variable).
Note also that we don't care about what these algorithms output. Their only
task is to achieve $\cost_\mu(Q;\phi)<0$ for our choice of $\phi$.
Given a cost function $c$ for $F$, what is the most expensive $c'$ for $F'$ so
that $\cost_\mu(Q;\phi)\ge0$ for any such decision tree?

%>>>
\paragraph{The quantity $P(k,n)$.}%<<<

Consider the $\th nk$ function. Looking into the definition of cost function
$\phi$ above, note that it charges an algorithm a positive value for reading
a 0 and a negative value for a 1. Furthermore, the greater the negative value
is, the more expensive $c'$ is.
We are interested in the values defined below.

\begin{definition}\label{problem}
	Let $\nu_{k,n}$ be the uniform distribution over 
	${[n]\choose k}$.
	For $\eta\in\R_\ge$ and $0<k\le n$, 
	define a local cost-function $c_\eta=(1,-\eta)$, 
	\[
		P_\eta(k,n)
			=\min_Q\bigl\{\cost_{\nu_{k,n}}(Q;c_\eta))\bigr\}\qand
		P(k,n)
			%&=\max\bigl\{\eta:(\forall Q)(\cost_{\nu_{k,n}}(Q;c_\eta)\ge0)\bigr\}\\
			=\max\bigl\{\eta:P_\eta(k,n)\ge0\bigr\}
	,\]
	where $Q$ ranges over all non-empty decision trees.
	%Define also $P(0,n)=\infty$.
\end{definition}

The relevance of this definition is demonstrated by Proposition~\ref{prop} below.

\begin{remark}
	We note that, for any given algorithm $Q$, its cost is a linear function of
	$\eta$: $\cost_{\nu_{k,n}}(Q;c_\eta)=A-\eta B$, where $A$ is the expected
	number of 0s and $B$ the expected number of 1s that $Q$ reads under
	distribution $\nu_{k,n}$.

	It follows that $P_\eta(k,n)$ is continuous as a function of $\eta$
	and so for $\eta=P(k,n)$ we must have $P_\eta(k,n)=0$.
	Furthermore, since we consider non-empty algorithms only, $P_\eta(k,n)$ is
	decreasing as a function of $\eta$.
\end{remark}

We call an algorithm $Q$ with $\cost_{\nu_{k,n}}(Q;c_\eta)=P_\eta(k,n)$,
an {\em optimal algorithm for $P_\eta(k,n)$};
if $\eta=P(k,n)$, an {\em optimal algorithm for $P(k,n)$}.

The normalized local cost-function $c_\eta$ charges an algorithm 1 for reading
a 0 and pays the algorithm $\eta$ when it discovers a 1. 
What is the greatest value of $\eta$ for which
$\cost_{\nu_{k,n}}(Q;c_\eta)\ge0$ for any $Q$?

We now discuss how lower bounds for $P(k,n)$ and $P(n-k+1,n)$ lead to lower
bounds for threshold read-once functions.
Suppose $c_0(i)=c_0$ and $c_1(i)=c_1$, for all $i\in[n]$, and write
$c'_1=c'_1(u)$ and $c'_0=c'_0(u)$.
If $\alpha\le P(k,n)$ and $\beta\le P(n-k+1,n)$, 
then we may set
\begin{equation}\label{eq:costdef}
	{c_1'\choose c_0'}=
	\left(
		\begin{array}{cc}
			k     				& \alpha k \\
			\beta(n-k+1) 	& n-k+1  	 \\
		\end{array}
	\right)
  {c_1\choose c_0}
.\end{equation}
\begin{proposition}\label{prop}
	$\cost_{\mu}(Q;c)\ge\cost_{\mu'}(Q';c')$.
	%$\cost_\mu(Q,c-\psi)\ge0$.
\end{proposition}
\begin{proof}
	Recall the definition of cost function $\phi$ and the related
	expression~(\ref{eq:phi}).
	Note that $\phi$ charges $0$ for $i\notin[n]$. For $i\in[n]$, $\phi$ charges
	$c_0$ for reading a $0$ and $-\alpha c_0$ for reading a $1$ over the support
	of $\nu_{k,n}$ and, respectively, $-\beta c_1$ and $c_1$ over the support of
	$\nu_{k-1,n}$.
	An assignment on the variables outside $[n]$, determines an algorithm over
	$\{0,1\}^n$.
	The conditional distribution of the variables in $[n]$,
	given such an assignment and a value for $u$, 
	is either $\nu_{k,n}$ or $\nu_{k-1,n}$.
	It follows that $\cost_\mu(Q;\phi)$ is lower-bounded by a convex
	combination of $c_0P_\alpha(k,n)\ge0$ and $c_1P_\beta(n-k+1,n)\ge0$.
\end{proof}

Applying this process repeatedly, shrinking all sibling leaves to their
parent, we reduce a tree of depth $d$ to one of depth $d-1$. After $d$ such
steps and starting with a local cost-function that charges $c_1$ and $c_0$ for
reading 1 and 0 respectively, we are left with a single node and the cost
function defined by
\begin{equation}\label{eq:costroot}
  {c'_1\choose c'_0}=\Gamma_{k,n}^d{c_1\choose c_0},~~\text{where}~~
	\Gamma_{k,n}=\left(
		\begin{array}{cc}
			k     				& \alpha k \\
			\beta(n-k+1) 	& n-k+1  	 \\
		\end{array}
	\right)
.\end{equation}
We will obtain a lower bound in the order of $\lambda^d$, where $\lambda$ is
the largest eigenvalue. 
Denoting the trace of $\Gamma_{k,n}$ by $T$ and its determinant by $D$,
\begin{equation}\label{eq:eigen}
	\lambda
		=\frac T2+\sqrt{\frac{T^2}4-D},\quad
		\text{where}~~T=n+1~~\text{and}~~D=(1-\alpha\beta)k(n-k+1)
.\end{equation}
To accomplish this, we will set $c_1=1$ and $c_0$ so that ${c_1\choose c_0}$
is an eigenvector for $\lambda$.
The matrices considered are positive and it can be shown using
a theorem of Perron (see Theorem~8.2.2 in \cite{hj}) that $c_0$ is positive.
As long as $c_0$ is also $\Theta(1)$,
our asymptotic bounds are not affected. We note for reference that
for a matrix ${a~b\choose c~d}$ with $b\ne0$, ${1\choose(\lambda-a)/b}$ is an
eigenvector for its greatest eigenvalue $\lambda$.
%>>>
\subsection{Properties and estimations of $P(k,n)$}%<<<
Given $(\ref{eq:eigen})$, we seek a lower bound of $P(k,n)$.
We start with some monotonicity properties of $P(k,n)$.
Note that, by symmetry, we may consider algorithms that read the variables
in the order $x_1,x_2,\dots,x_n$.

Intuitively, if we increase $n$, we should be able to increase $\eta$ as it 
becomes harder for an algorithm to discover a 1. This is captured in the next
statement.
\begin{theorem}\label{thm:monotone}
	For any $\eta>0$ and $0<k\le n$, $P_\eta(k,n)<P_\eta(k,n+1)$.
\end{theorem}
\begin{proof}
	Consider any algorithm $Q$ over $\{0,1\}^{n+1}$ and define a randomized algorithm $Q'$
	over $\{0,1\}^n$ as follows. Algorithm $Q'$ on input
	$(x_1,\dots,x_n)$, chooses $i\in[n+1]$ and simulates $Q$ on input
	$(x_1,\dots,x_{i-1},0,x_i,\dots,x_n)$. We claim that the cost of $Q'$ is
	less than the cost of $Q$. This is not hard to see as $Q'$ will gain $-\eta$
	whenever $Q$ queries a 1 and also gains $-\eta$, but with positive
	probability (when $i=1$ for example) will pay 0 when $Q$ queries $i$ and
	pays 1. 
\end{proof}

Similarly, increasing both $k$ and $n$ by $1$, and since $k/n<(k+1)/(n+1)$,
one might expect that we would need to decrease $\eta$.

\begin{theorem}\label{thm:monotone2}
	For any $\eta>0$ and $0<k\le n$, $P_\eta(k,n)>P_\eta(k+1,n+1)$.
\end{theorem}
\begin{proof}
	Consider any algorithm $Q$ over $\{0,1\}^n$ and 
	define a randomized algorithm $Q'$ over $\{0,1\}^{n+1}$ as follows.
	Algorithm $Q'$ chooses $i\in[k+1]$ and simulates $Q$, but when it encounters
	the $i$-th 1 while querying a variable $j$, $Q'$ continues with the rest of
	the variables as if $x_j$ was skipped or missing. 

	More formally, for $i\in[k+1]$, let $Q_i$ be the algorithm which results
	from $Q$ as follows.
	Consider a path $(u_1,x_1=b_1,u_2,x_2=b_2,\dots,u_j,x_j=1,u_{j+1})$ in $Q$,
	such that $x_j=1$ is the $i$-th query that returned 1 (i.e., exactly $i$ of
	the bits $x_1,\dots,x_j$ are 1 and $x_j=1$), and let $T$ denote the sub-tree
	rooted at node $u_j$ of $Q$ that queries $x_j$.
	Let $T'$ be identical to $T$, but with every query shifted by 1; that is,
	every node of $T$ labeled by $\ell\in[n]$ is replaced with $\ell+1\in[n+1]$.
	Let $Q_i$ be the tree which results from $Q$ by replacing the (possibly
	empty) sub-tree rooted at node $u_{j+1}$ with $T'$ and applying this
	transformation for every such path in $Q$.
	Call the edge $x_j=1$ a \emph{skipped} edge.
	Finally, $Q'$ is the uniform
	distribution over the $Q_i$'s.

	It suffices to show that $Q'$ has strictly smaller cost than $Q$. Consider
	a path $P$ in $Q$ of length $m$ with $\ell$ 1's. The probability $P$ is
	executed in $Q$ is ${n-m\choose k-\ell}/{n\choose k}$.
	There is a natural map of a path in $Q_i$ to a path in $Q$: ignore a skipped
	edge (if it exists) and shift back by one any shifted label.
	%Regarding paths in $Q'$, we will identify them with paths of $Q$ by removing
	%from $Q_i$ the edge corresponding to an $i$-th 1.
	We compute the probability of the paths in $Q'$ that map to $P$ by counting
	the number of pairs $(i,x)$, such that the execution path of $Q_i$ on input
	$x$ is mapped on $P$.
	\[
		\biggl[m{n-m\choose k-\ell}+(k+1-\ell){n+1-m\choose k+1-\ell}\biggr]
			\bigg/\biggl[(k+1){n+1\choose k+1}\biggr]
	.\]
	The second term of the sum is the number of such pairs without skipped
	edges; equivalently, with $i>\ell$ and the first $m$ bits of $x$ consistent
	with $P$.
	The first term counts the rest of the pairs $(i,x)$, since there are $m$
	choices for the skipped 1 and this choice determines $i$.
	This probability can be verified to equal ${n-m\choose k-\ell}/{n\choose k}$.
	%\[
	%	\left[m{n-m\choose k-\ell}+(k+1-\ell){n+1-m\choose k+1-\ell}\right]
	%		\Bigg/\left[(k+1){n+1\choose k+1}\right]
	%	=
	%	{n-m\choose k-\ell}\Bigg/{n\choose k}
	%\]
	Furthermore, since $Q$ is non-empty, there is at least one skipped edge in
	$Q_1$ taken with positive probability, implying $Q'$ has strictly
	smaller cost. % than $Q$, as desired.
\end{proof}

\begin{corollary}\label{cor:stopat1}
	%The cylinder $(1,*,\dots,*)$ belongs to any optimal decision tree for $P(k,n)$.
	An optimal algorithm for $P(k,n)$ stops if the first query reads 1.
\end{corollary}
\begin{proof}
	Consider an optimal tree for $P(k,n)$. 
	By Theorem~\ref{thm:monotone2}, 
	for $\eta=P(k,n)$, we have $P_\eta(k-1,n-1)>P_\eta(k,n)=0$. 
	Thus, descending further after reading a 1 will add a positive expected cost.
\end{proof}

Using the above property of optimal algorithms, we can determine $P(1,n)$ and
$P(n,n+1)$.
\begin{lemma}\label{lemma:p1n}
	For all $n>1$, $P(1,n)=(n-1)/2$.
\end{lemma}
\begin{proof}
	Theorem~\ref{thm:monotone} determines the optimal decision tree $Q$ for
	$P(1,n)$: it queries the variables one by one until a 1 is read.
	This is because, for $\eta=P(1,n)$, by the theorem $P_\eta(1,n-1)<0$; thus,
	stopping after a 0 is read is not optimal.
	For $\eta>0$,
	\[
		\cost_{\nu_{1,n}}(Q;c_\eta)
		=\sum_{i=1}^n\frac1n\cdot(i-1-\eta)
		=\frac{n-1}2-\eta
	.\]
	It follows that $\cost_{\nu_{1,n}}(Q;c_\eta)=0$ if and only if
	$\eta=(n-1)/2$.
\end{proof}
\begin{lemma}\label{lemma:plus1}
	For all $n>1$, $P(n,n+1)=1/(2n)$.
\end{lemma}
\begin{proof}
	By Corollary~\ref{cor:stopat1}, the optimal (non-empty) tree stops if
	$x_1=1$ and if $x_1=0$ it proceeds to discover all 1s. Its cost, for
	$\eta>0$, is $-\eta n/(n+1)+(1-\eta n)/(n+1)$.
	This equals zero if and only if $\eta=1/(2n)$.
\end{proof}

We can now prove the following lower bound on $P(k,n)$.

\begin{theorem}
	For $0<k<n$, $\frac{n-k}{2k}\le P(k,n)$.
	%and the inequality is strict unless $k=1$ or $k=n-1$.
\end{theorem}
\begin{proof}
	The intuition of the inductive proof below is that if $P(k,n)$ is smaller
	than $(n-k)/(2k)$,
	then $P_\frac{n-k}{2k}(k,n-1)$ should be absurdly small.

	Note first that for any $k>0$ and $n=k+1$ the statement follows (with
	equality) from Lemma~\ref{lemma:plus1}.
	Suppose---towards a contradiction---that for a lexicographically least pair
	$(k,n)$ with $n>k+1$, $P(k,n)<(n-k)/(2k)$. 
	Thus, 
	\begin{equation}\label{eq:ineq}
		P_\frac{n-k}{2k}(k,n-1)<P_\frac{n-k}{2k}(k,n)<0
	,\end{equation}
	where the first inequality follows from Theorem~\ref{thm:monotone} and the
	second from our assumption that $P(k,n)<(n-k)/(2k)$. Furthermore, by the
	minimality of $(k,n)$,
	\[
		P(k-1,n-1)\ge\frac{n-k}{2(k-1)}>\frac{n-k}{2k}
	,\]
	which in turn implies
	\[
		0<P_\frac{n-k}{2k}(k-1,n-1)
	.\]
	This last inequality implies that an optimal decision tree for
	$P_\frac{n-k}{2k}(k,n)$ will not descent further after reading $x_1=1$. 
	We conclude that
	\begin{equation}\label{eq:cost}
		P_\frac{n-k}{2k}(k,n)
		=\frac{n-k}n\cdot\Bigl[1+P_\frac{n-k}{2k}(k,n-1)\Bigr]-\frac{n-k}{2k}\cdot\frac kn
	.\end{equation}
	Suppose $Q$ is an optimal algorithm for $P_\frac{n-k}{2k}(k,n-1)$. There are
	constants $u$ and $v$ such that, for any $\eta>0$,
	\[
		\cost_{\nu_{k,n-1}}(Q;c_\eta)=u-\eta\cdot v
		%\cost_{\nu_{k,n-1}}(Q;c_\frac{n-k}{2k})=u-\frac{n-k}{2k}\cdot v
		%\hbox{~~and~~}
		%\cost_{\nu_{k,n-1}}(Q;c_\frac{n-1-k}{2k})=u-\frac{n-1-k}{2k}\cdot v
	.\]
	Combining (\ref{eq:ineq}) and (\ref{eq:cost}) and substituting
	the above expression with $\eta=(n-k)/(2k)$, we obtain after simplification
	\[
		u-\frac{n-k}{2k}\cdot v<-\frac12
	.\]
	On the other hand, by the minimality of $(k,n)$, 
	$P(k,n-1)\ge(n-1-k)/(2k)$. This implies
	\[
		u-\frac{n-1-k}{2k}\cdot v\ge0
	.\]
	The above two inequalities are satisfiable only if $v>k$. However, $v$ is
	the expected number of 1s read by $Q$, which is at most $k$;
	a contradiction.
\end{proof}
%>>>

% Section 4
\section{Bounds for uniform read-once functions}%<<<
We begin with an analysis of the (straightforward)
directional algorithm. We then show that this
algorithm is optimal for uniform AND-OR trees.
Finally, we obtain bounds for uniform threshold read-once functions.
%>>>
\subsection{Analysis of the directional algorithm}% Subsection 4.1 <<<
A simple directional algorithm evaluates a node by evaluating its children in
a randomly chosen order. To analyze this algorithm on a reluctant input, let
$\Phi_d$ denote the expected cost to evaluate a node of height
$d$ when the output is 1 and $\Psi_d$ when it is 0. Consider a node with
exactly $k$ of its children evaluating to 1.
The number of its children that will be evaluated until its value is
determined by the directional algorithm follows a negative hypergeometric
distribution. 
The expected number of children that will be evaluated to 0 is $(n-k)k/(k+1)$.
(An elementary exposition of this distribution and calculation of its mean can
be found in \cite[Chapter~7]{ross}.)
Similarly, when dealing with a node with exactly $n-k+1$ of its children
evaluating to 0, the expected number of children that will be evaluated to
1 is
$(k-1)(n-k+1)/(n-k+2)$.
Thus, along with $\Phi_0=c_0$ and $\Psi_0=c_1$,
we obtain the recurrence
\begin{equation}\label{eq:recurrence}
	{\Phi_d\choose\Psi_d}=\Delta_{k,n}{\Phi_{d-1}\choose\Psi_{d-1}},
	~~\text{where}~~
	\Delta_{k,n}=
	\left(\begin{array}{cc}
		k 														& (n-k)\cdot\frac k{k+1}  \\
		(k-1)\cdot\frac{n-k+1}{n-k+2} & n-k+1										\\
	\end{array}\right)
.\end{equation}
%>>>
\subsection{Uniform alternating AND-OR trees} %<<<

In this subsection we prove Theorem~\ref{thm:andor}, by showing that the 
directional algorithm for the alternating AND-OR tree satisfies the same
recurrence as the one we obtain with the method of generalized costs. In
particular, we show that the recurrences (\ref{eq:costdef}) and
(\ref{eq:recurrence}) coincide for uniform alternating AND-OR
trees.

Let us define the analogous quantities
$\Phi^\wedge_d$ and $\Psi^\wedge_d$ (resp.~$\Phi^\vee_d$, $\Psi^\vee_d$) for
the expected cost when evaluating a node of height $d$ and
labeled by an AND (resp.~OR) gate to 1 and 0 respectively.
Recall Equation~(\ref{eq:costroot}) and note that 
$\Delta_{n,n}=\Gamma_{n,n}$ and
$\Delta_{1,n}=\Gamma_{1,n}$, since $P(n,n)=0$ and $P(1,n)=(n-1)/2$
(Lemma~\ref{lemma:p1n}).
Thus, Equation~(\ref{eq:recurrence}) instantiates to
\begin{equation}\label{eq:andorrec}
	{\Phi^\wedge_d\choose\Psi^\wedge_d}=
		%\left(\begin{array}{cc}
		%	n 			& 0 \\
		%	(n-1)/2 & 1	\\
		%\end{array}\right)
		\Gamma_{n,n}
		{\Phi^\vee_{d-1}\choose\Psi^\vee_{d-1}}
	\quad\text{and}\quad
	{\Phi^\vee_d\choose\Psi^\vee_d}=
		%\left(\begin{array}{cc}
		%	1 & (n-1)/2 \\
		%	0 & n	\\
		%\end{array}\right)
		\Gamma_{1,n}
		{\Phi^\wedge_{d-1}\choose\Psi^\wedge_{d-1}}
.\end{equation}
Defining $A=\Gamma_{n,n}$ and $B=\Gamma_{1,n}$ we have
\begin{equation}\label{eq:andor}
	{\Phi^\wedge_{2d}\choose\Psi^\wedge_{2d}}=(AB)^{d}{c_1\choose c_0}
		\qand
	{\Phi^\vee_{2d}\choose\Psi^\vee_{2d}}=(BA)^{d}{c_1\choose c_0}
.\end{equation}

With respect to the lower bound,
in view of Equations~(\ref{eq:costdef}), (\ref{eq:costroot}), and
(\ref{eq:eigen}),
we obtain that if there is a decision tree $Q$ for $F^{2d}_{\wedge,n}$ of cost
$C$, then there is a decision tree for a single variable of cost $C$ under the
cost-function
\[
	{c'_1\choose c'_0}=(AB)^{d}{c_1\choose c_0},\quad\text{where}~~
	AB=\left(
		\begin{array}{cc}
			n 					& \frac{n(n-1)}2	  \\
			\frac{n-1}2 & \frac{(n+1)^2}4 \\
		\end{array}
	\right)
.\]
Comparing the above with Equation~(\ref{eq:andor}), the first
part of Theorem~\ref{thm:andor} follows.

To obtain the asymptotic bound, observe that the trace $T$ and the determinant
$D$ of $AB$ are $T=n+{(n+1)^2}/4$ and $D = n^2$.
Let $c_1=1$ and $c_0=(\lambda-n)/(n(n-1)/2)$, so that ${c_1\choose c_0}$ is an
eigenvector for the greatest eigenvalue $\lambda$ of $AB$, where
\[
	\lambda
	=n+\frac{(n-1)^2}8+\frac{(n-1)^2}8\sqrt{1+\frac{16n}{(n-1)^2}}
	=\Biggl[\frac{n-1}4\Biggl(1+\sqrt{1+\frac{16n}{(n-1)^2}}\Biggr)\Biggr]^2
.\]
Since $c_0$ is $\Theta(1)$ and positive, we obtain
$R(F^{2d}_{\land,n})=\Theta(\lambda^d)$.
Bounds for trees of odd depth, if desired, can be obtained via the recurrence
relations (\ref{eq:andorrec}).
%>>>
\subsection{Uniform threshold read-once functions}%<<<

With respect to the upper bound for $F^d_{k,n}$, we obtain directly from (\ref{eq:recurrence})
\[
	{\Phi_d\choose\Psi_d}=
	\Delta^d_{k,n}
	{c_1\choose c_0}
.\]
We set $c_1=1$ and $c_0=(\lambda-k)/\bigl[(n-k)\frac k{k+1}\bigr]$, so that
${c_1\choose c_0}$ is an eigenvector for the greatest eigenvalue $\lambda$ of
$\Delta_{k,n}$, where
\[
	\lambda
	=\frac{n+1}2\left(1+\sqrt{1-\frac{8k(n-k+1)}{(n-k+2)(k+1)(n+1)}}\right)
.\]
Since $c_0=\Omega(1)$, Theorem~\ref{thm:directional} follows.

We now turn to the lower bound.
From the bounds on $P(k,n)$ we obtained in the preceding section it follows we
can apply the method of generalized costs with
\[
  \Gamma_{k,n}=
	\left(
		\begin{array}{cc}
			k     				& \frac{n-k}2 \\
			\frac{k-1}2 	& n-k+1  	 \\
		\end{array}
	\right)
.\]
Its trace and determinant are
$T=n+1$ and $D=k(n-k+1)-(k-1)(n-k)/4$
respectively. Again,
with $\lambda$ the largest eigenvalue,
\begin{align*}
	\lambda
		&=\frac{n+1}2+\frac{n+1}2\sqrt{1-\frac{3k(n-k+1)+n}{(n+1)^2}}
,\end{align*}
we set $c_1=1$ and $c_0=(\lambda-k)/\frac{n-k}2$.
Since $(c_1,c_0)$ is an eigenvector for $\lambda$ and $c_0=O(1)$,
$R(F^d_{k,n})=\Omega(\lambda^d)$ and we have obtained Theorem~\ref{thm:main}.
%>>>

% Section 5
\section{Future directions}%<<<
The main motivation for this work has been the recursive majority-of-three
function. Although the value $P(2,3)$ leads to the known $2.5^d$ lower bound,
one can define an analogous problem on trees of some height $t$.
Informally, $P^t(2,3)$ could ask for the greatest value $\eta$ so that any
algorithm for $\maj_t$ would have a non-negative cost under $c_\eta$.
Such a problem can be studied either rigorously or computationally (say using
an LP-based approach as in \cite{gj2016}). 

Another direction is to generalize the current work to unbalanced functions,
in the spirit of \cite{hw91} or \cite{hnw93}. A first step in this direction
would be to obtain the results of \cite{amano2011} using the generalized-costs
method.
%>>>

\bibliographystyle{plain}
\bibliography{threshold}	
\end{document}